\documentclass[sigconf]{acmart}

\usepackage{booktabs} 
\usepackage{algorithm}
\usepackage{algpseudocode}
\usepackage{algorithmicx}
\usepackage{algcompatible}
\usepackage{float}
\usepackage{cleveref}
\usepackage{setspace}
\setcopyright{rightsretained}

\acmDOI{10.475/123_4}

\acmISBN{123-4567-24-567/08/06}

\acmConference[CIKM'17]{ACM Conference on Information and Knowledge Management}{November 2017}{Singapore} 
\acmYear{2017}
\copyrightyear{2017}

\acmPrice{15.00}

\renewcommand\footnotetextcopyrightpermission[1]{} 
\begin{document}
\title{Distributed Training Large-Scale Deep Architectures}

 \author{Shang-Xuan Zou}
 \affiliation{
   \institution{HTC AI Research}
   \city{Taipei} 
   \state{Taiwan} 
 }
 \email{mina\_zou@htc.com}
 
 \author{Chun-Yen Chen}
 \affiliation{%
   \institution{HTC AI Research}
   \city{Taipei} 
   \state{Taiwan} 
 }
 \email{arbit\_chen@htc.com}
 
 \author{Jui-Lin Wu}
 \affiliation{%
   \institution{HTC AI Research}
   \city{Taipei} 
   \state{Taiwan} 
 }
 \email{ruilin\_wu@htc.com}
 
 \author{Chun-Nan Chou}
 \affiliation{
   \institution{HTC AI Research}
   \city{Taipei} 
   \state{Taiwan} 
 }
 \email{Jason.CN\_Chou@htc.com}
 
 \author{Chia-Chin Tsao}
 \affiliation{
   \institution{HTC AI Research}
   \city{Taipei} 
   \state{Taiwan} 
 }
 \email{dustin\_tsao@htc.com}
 
 \author{Kuan-Chieh Tung}
 \affiliation{
   \institution{HTC AI Research}
   \city{Taipei} 
   \state{Taiwan} 
 }
 \email{jacky\_tung@htc.com}
 
 \author{Ting-Wei Lin}
 \affiliation{%
   \institution{HTC AI Research}
   \city{Taipei} 
   \state{Taiwan} 
 }
 \email{beck\_lin@htc.com}
 
 \author{Cheng-Lung Sung}
 \affiliation{
   \institution{HTC AI Research}
   \city{Taipei} 
   \state{Taiwan} 
 }
 \email{cl\_sung@htc.com}
 
 \author{Edward Y. Chang}
 \affiliation{
   \institution{HTC AI Research}
   \state{USA} 
 }
 \email{edward\_chang@htc.com}
 
\begin{abstract}
Scale of data and scale of computation infrastructures together enable the current deep learning renaissance.
However, training large-scale deep architectures demands both algorithmic improvement and careful system configuration.
In this paper, we focus on employing the system approach to speed up large-scale training.
Via lessons learned from our routine benchmarking effort,
we first identify bottlenecks and overheads that hinter data parallelism. 
We then devise guidelines that help practitioners to configure an effective system and fine-tune parameters to achieve desired speedup. 
Specifically, we develop a procedure for setting mini-batch size and choosing computation algorithms.  We also derive lemmas for determining the quantity of key components such as the number of GPUs and parameter servers. Experiments and examples show that these guidelines help effectively speed up large-scale deep learning training.
\end{abstract}

%
%

\ccsdesc[500]{Computing methodologies~Batch learning}
\ccsdesc[500]{Computing methodologies~Neural networks}
\ccsdesc[500]{Computing methodologies~Parallel computing methodologies}
\ccsdesc[500]{Computing methodologies~Distributed computing methodologies}


\keywords{Deep learning, Neural network, Convolutional neural networks, Distributed learning, Speedup, Performance tuning}

\maketitle

\section{Introduction} \label{sec-intro}
In the last five years, neural networks and deep architectures have been proven very effective in application areas such as computer vision, speech recognition, and machine translation.
The recent breakthroughs of AlphaGo further cement interest in employing deep architectures to develop intelligent machines. 
Although deep architectures such as convolutional neural networks (CNNs)~\cite{LeCun:1998ek, NIPS2012_alexnet, Graves:2013wt}, recurrent neural networks (RNNs)~\cite{Zaremba:2014up, Graves:2014vz}, and restricted Boltzman machines (RBMs)~\cite{fischer2012introduction, Krizhevsky:2010va} have been around since the 1980s, they have never been under the spotlight. Why are they thriving now? The convincing factor this time around is \emph{scale}, in both data volume and computation resources.

When the scale of training data is small, all supervised learning algorithms (e.g., decision trees, support vector machines, and logistic regression) achieve the same level of classification accuracy. In 2012, AlexNet~\cite{NIPS2012_alexnet} demonstrated that with millions of training images from ImageNet~\cite{imagenet_cvpr09}, CNNs substantially outperform all prior works on image classification. Since then it has been shown in several vertical domains that large training datasets can improve the accuracy of classification tasks.

Since the computation complexity of a deep learning algorithm is high (e.g., the convolution stage of CNNs requires a six-level nested loop), the scale of data demands scalable computation resources. Fortunately, processor speed has soared more than one thousand folds in the last three decades. In addition, with specialized arrays of processors (e.g., GPUs) and accessibility of parallel computing infrastructures via the cloud, millions of cores can be utilized simultaneously for training. However, scaling up computation is not merely throwing in an infinite number of cores.
As Amdahl's law \cite{amdahl1967validity} states, the non-parallelizable portion of a computation task such as communication, I/O, and interprocess synchronization may cap computation speedup.
For instance, if the non-parallelizable portion is $50\%$, reducing computation time to zero achieves only a speedup factor of two. All deep learning frameworks involve substantial non-parallelizable overheads, which must be carefully mitigated to speed up training time.

Several open-source projects (e.g., Caffe \cite{caffe2014}, MXNet \cite{chen2015mxnet}, TensorFlow \cite{abadi2015tensorflow}, and Torch \cite{collobert2011torch7}) have been devoted to speeding up training deep networks.  They can be summarized into two approaches: deep-learning algorithm optimization and algorithm parallelization (details of related work are presented in Section~\ref{related work}). The former includes using better convolution algorithms, improving stochastic gradient decent (SGD) with faster methods, employing compression/quantization, and tuning the learning rate with advanced optimization techniques. Indeed, most open-source libraries have quickly adopted available state-of-the-art optimizations. 
However, most users in academia and industry do not know how to set parameters, algorithmic and system, to conduct cost-effective training.  Researchers and professionals face at least the following questions in three levels, which are intra-GPU, inter-GPU, and inter-machine:

\begin{enumerate}
\item What is the bottleneck of speeding up deep learning training by parallelism?
\item With $X$ amount of data, what is the size of each mini-batch ($X_{mini}$) and how to maximize GPU utilization?
\item How many GPUs ($G$) should be employed, and how should such a system be configured?
\item How many parameter servers ($N_{ps}$) should be deployed when building a distributed system?
\end{enumerate}

In this work, we aim to answer the above questions by providing system configuration guidelines given the characteristics of the training data (the number of training instances and the size of each training instance), as well as 
hardware parameters (such as GPU memory size, internal transmission bandwidth, e.g. bus bandwidth, and external transmission bandwidth, e.g. network bandwidth). 
We identify computation bottlenecks and I/O overheads of representative frameworks. From the insights we observed in benchmarking, 
we propose guidelines allowing users to configure a high-performance deep learning system for their target tasks. 

\subsection{Related Work}
\label{related work}
\begin{figure*}[!ht]
\centering
\includegraphics[width=0.9\textwidth]{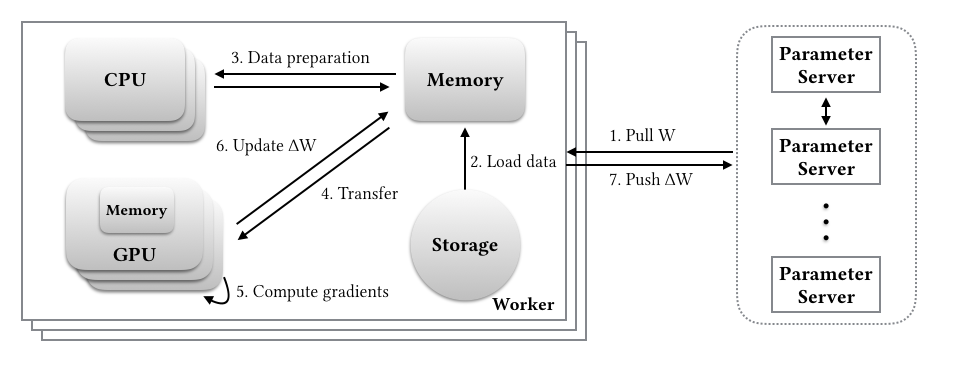}
\vspace{-2em}
\caption{Deep learning system architecture.
The batch processing pipeline in the general training process can be divided into seven steps. Each of them has its own factors that influence training performance.}
\vspace{-1em}
\label{fig:sec3_node_arch}
\end{figure*}

Since deep-learning training is time-consuming, many previous studies devoted to speeding up the training performance.
These prior contributions can be divided into two approaches: algorithmic and system.  The algorithmic approach accelerates the training algorithm, whereas the system approach focuses on employing improved resources to achieve parallel training. 
To ensure scalability, the system approach may require enhancing the training algorithm to take full advantage of the increased resources.

\subsubsection{Algorithmic Approach}

Stochastic gradient descent (SGD) is the de facto optimization algorithm for training a deep architecture. Many SGD techniques have been developed for achieving faster convergence to the global minimum. The settings of hyper-parameters such as learning rate and mini-batch size are crucial to the training performance.
Hinton and Bengio~\cite{hinton2010practical, bengio2012practical} provide recommendations on setting hyper-parameters commonly used in gradient-based training. 
Batch renormalization can be an effective strategy to train a network with small or non-i.i.d mini-batches~\cite{Ioffe:2017uc}.
Momentum-based acceleration schemes increase the speed of learning and damp oscillations in directions of high curvature~\cite{polyak1964some}. 
Per-parameter adaptive learning rate methods help reduce large gradients and decrease the learning rate over time~\cite{duchi2011adaptive}. 

More efficient algorithms can improve speed.
The execution time of convolution consumes $70\%$ to $90\%$ of CNN-based training. Some FFT-based convolution schemes were proposed~\cite{Mathieu:2013wa} to achieve speedup. Additionally, Firas \emph{et al.} proposed three matrix layout schemes using lowering operations~\cite{Hadjis:2015wx}. \emph{Caffe con Troll} implements a CPU-GPU hybrid system that contains several lowering operations, and at the same time, employs a simple automatic optimizer to select the best lowering. Some compression algorithms~\cite{elgohary2016compressed} are developed for both good compression ratios and fast decompression speed to enable block-wise uncompressed operations, such as matrix multiplication are executed directly on the compressed representations. 
 
\subsubsection{System Approach}

A deep learning training job consists of two computationally intensive arithmetic operations: matrix multiplication and convolution. A GPU is well-suited for speeding up such operations since these operations are easy to be parallelized. To achieve further speedup, the next logical step is to employ multiple GPUs, and to configure a distributed clusters of CPUs and GPUs. The computation time can be largely reduced via data parallelism and/or model parallelism. Many projects have proven parallelism to be helpful \cite{chilimbi2014project, dean2012large, krizhevsky2014one, niu2011lock, Iandola:2016in,zhang2015poseidon}.

According to Amdahl's law, the peak performance of a parallel architecture is capped by the overhead portion of the computation task. In the context of deep learning, its training overhead includes synchronization between distributed threads, disk I/O, communication I/O, and memory access. 
To reduce synchronization delay, Zinkevich \emph{et al.} \cite{zinkevich2010parallelized} proposed an asynchronous distributed SGD algorithm to guarantee parallel acceleration without tight latency constraints. Chen \emph{et al.} \cite{chen2016revisiting} proposed adding backup workers in synchronous SGD algorithm to mitigate the bottleneck. To reduce the impact of I/O on the overall speedup, most open-source frameworks (see Section 1.1.3) attempt to conceal I/O behind computation via the pipeline approach proposed in \cite{PLDA+}.  Such approach requires a computation unit to be sufficiently long so as to hide I/O overheads as much as possible. The pipeline approach, however, demands carefully setting up the unit size of computation (or mini-batch size) and the number of parameter servers. We will propose how to best estimate these configuration parameters in Section \ref{sec:conf_perf}.

\subsubsection{Computation Frameworks}

There have been several deep learning open-source efforts. Here, we introduce representative frameworks\footnotemark:

\footnotetext {Due to limited information available, some frameworks, such as CNTK from Microsoft \cite{Dally:18pvCAH2} and Theano \cite{James:0hLkbVip}, are not covered.}  

\begin{itemize}

\item \textbf{Caffe}: Caffe \cite{caffe2014} is maintained and developed by the Berkeley Vision and Learning Center (BVLC) and has become open-source since 2014. Caffe was first designed for vision, and has been adopted and improved by users in several domain, such as speech recognition and robotics. In Caffe, some extensible toolkits are provided for state-of-the-art deep learning algorithms. Caffe separates network representation from actual implementation, and supports seamless switching between open-source platforms.

\item \textbf{MXNet}: MXNet \cite{chen2015mxnet} is designed for portability (i.e., supporting multiple languages and operating systems), scalability (i.e., running on multiple machines, GPUs and CPUs), and efficiency. Additionally, MXNet provides a superset programming interface to be compatible with other frameworks. MXNet is lightweight and it enjoys multiple programming language supports, e.g., Python, R, Julia and Scala.

\item \textbf{TensorFlow}: TensorFlow \cite{abadi2015tensorflow}, which supports distributed computation, is an open-source framework developed by Google. TensorFlow's design philosophy is flexibility, portability, and high efficiency. TensorFlow takes computations described by using a dataflow model and maps them onto a wide variety of hardware platforms. TensorFlow allows clients to easily express various kinds of parallelism through replication and parallel execution of a core model dataflow graph, with many different computational devices all collaborating to update a set of shared parameters or states.

\item \textbf{Torch}: Torch~\cite{collobert2011torch7} is designed to be easy for developing and extending numerical algorithms. Based on this philosophy, Torch leverages \emph{Lua} language, a fast interpreted language (with also the fastest Just In Time (JIT) compiler), to embedded in a C application and provides APIs in C, making library wrapping easily for the unifying interface to C/C++.
\end{itemize}
\vspace{0.5em}

Among the introduced frameworks, MXNet and TensorFlow are built-in distributed training frameworks. Users can easily develop algorithms running on computing clusters with thousands of CPUs or GPUs. Several works are proposed to give users a glimpse on the factors that they must take into consideration. Bahrampour \emph{et al.} \cite{Bahrampour2015uv} provide a comparative study on different frameworks with respect to extensibility, hardware utilization, and performance. Shi \emph{et al.}~\cite{shi2016} provides performance study on selected frameworks. These works offer practitioners a high-level guideline to select an appropriate framework. 
Given a selected framework, our work aims to provide further configuration guidelines to make training both fast and cost-effective.

\subsection{Contribution Summary}

In summary, this work makes the following contributions:
\begin{enumerate}
\item {\em Identifying computation bottlenecks and devising their remedies}.
We benchmark representative networks and datasets to identify the typical bottlenecks of large-scale training. 
We then devise remedies to reduce or mask computation overheads (I/O and communication) to improve training speed. 
\item {\em Quantifying remedies into an optimization model}.  We formulate our remedies into an optimization model to determine the optimal mini-batch size and carefully balance memory and speed tradeoffs so as to employ the fastest algorithms given the memory constraint. 
\item {\em Recommending distributed configuration involving multiple GPUs and parameter servers}.
When the workload cannot be handled by a single GPU or machine, we propose lemmas to recommend the number of GPUs and parameter servers to configure so as to achieve cost-effective speedup.
\end{enumerate}
Both real-world deployment and empirical studies attest our remedies to be very effective.

\section{Preliminaries}
This section presents a typical deep learning training process including performance factors and their relevant parameters. 
We then show the setup of the evaluation environment.

\subsection{Deep Learning Training Process}
\label{ssec-training-process}

Figure~\ref{fig:sec3_node_arch} depicts a general architecture of deep-learning training and data flow.   
A local architecture is basically a commodity computer equipped with $G$ GPUs.
When aiming to improve parallelism via a distributed architecture, a worker and a parameter server can be replicated into multiple copies connected by a network. 
The mini-batch processing pipeline in the training process consists of seven steps. 
After the model parameters $W$ and the data processing pipeline is initialized, the training process repeats until all training data is seen.

\begin{enumerate}
\item {\em Parameter refresh}. In distributed training, the latest copy of model parameters $W$ is pulled from parameter servers at the beginning of each mini-batch processing. $W$ is then loaded onto GPU memory. A distributed environment consists of $N_{w}$ workers and $N_{ps}$ parameter servers for managing shared parameters.
\item {\em Data loading}. A subset of the $X$ training instances called {\em mini-batch} of size $X_{mini}$ is loaded from the persistent storage to the main memory.  
\item {\em Data preparation}. $X_{mini}$ instances are transformed into the required input format. These instances may be augmented to mitigate the over-fitting problem and enrich sample diversity. 
\item {\em Host to GPU transfer}. The mini-batch 
is loaded onto the memory of a GPU. If $G$ GPUs are employed, $G$ different mini-batches are loaded onto $G$ GPUs.
\item {\em GPU processing}. Required computations including matrix multiplication and convolution are performed on $G$ GPUs for the gradients against the given mini-batch. 
\item {\em Parameter update}. The delta $\Delta W$ is derived from the gradients and applied to the previous version of $W$ in main or GPU memory.  
\item {\em Distributed update}. The parameter updates are sent to parameter servers when distributed machines are configured. 
\end{enumerate}
\vspace{-0.3em}
Among the seven steps, step $5$ performs computation, and the other steps that cannot be hidden behind step $5$ are considered as overheads.  The larger fraction of the time which those overhead steps take, the less effective parallelism can achieve.
Therefore, our tasks are minimizing overhead time and hiding overheads via pipelining as much as possible.
The remainder of this paper is to demonstrate how the following parameters can be carefully tuned to achieve such goals, organized into four sections. 
In section~\ref{ssec:singlegpu}, we provide a procedure to recommend a mini-batch size that leads to maximum training performance. 
Section~\ref{ssec:multigpu} provides an in-depth analysis on training in a multi-GPU environment.
We provide a lemma to estimate the number of GPUs $G$ for a desired factor of speedup.
The increase of GPU number not only improves performance speedup, but also induces communication overheads between GPUs. We'll also discuss how to alleviate the impacts of these overheads.
In section~\ref{ssec:multihost}, we address issues involving distributed workers.
When the training system scales horizontally, we need an extra cluster to manage the parameters in addition to training hosts in the distributed environment.
The communication between training hosts and parameter servers is an overhead that could seriously degrade training speedup. 
We propose a scheme to estimate the number of parameter servers $N_{ps}$ given network bandwidth $B_{ps}$.

\subsection{Evaluation Environment}
\label{benchmark_environment}

We set up our evaluation environment with Elastic Compute Cloud (EC2) of Amazon Web Services (AWS)\footnotemark. \footnotetext{ GPU instances on Google Compute Engine (GCE) do not support GPU peer-to-peer access, and hence we will defer our GCE experiments till such support is available.}
All experiments run on EC2 P2 instances equipped with NVIDIA Tesla K80 Accelerators which contain a pair of NVIDIA GK210 GPUs.
Each GPU provides $12$ GB memory and $2,496$ parallel processing cores. 
The CPU is a customized version of Intel Broadwell processor running at $2.7$ GHz.
Table~\ref{p2instance} shows hardware configurations of P2 type instances\footnotemark.
\footnotetext{p2.16xlarge is not used in our experiments because it does not support full GPU-to-GPU communication which introduces one more variable in our multi-GPU experiments.}
To avoid unexpected GPU clock rate adjustment in our experiments, we disable GPU autoboost function. 

\begin{table}[ht]
\centering
\caption{AWS P2 instances}
\label{p2instance}
\begin{tabular}{|l|c|c|c|c|}
\hline
 Instance   & \#GPU  & 	GPU Mem.  & 	Network  	\\ \hline 
p2.xlarge  &	1	&	12 GB 		&	High 		    \\ \hline
p2.8xlarge &	8  	&	96 GB 		&	10 Gbps   	    \\ \hline
p2.16xlarge&	16	& 	192 GB	 	&	20 Gbps  	  	    \\ \hline
\end{tabular}
\end{table}

We perform experiments and demonstrate our ideas by MXNet and TensorFlow.
Virtual machines are launched from Amazon deep learning AMI (Amazon Machine Image) $v2.1$ preloaded with NVIDIA CUDA toolkit $v7.5$ and cuDNN $v5.1$. We conduct experiments on the ILSVRC-2012 dataset, the subset of ImageNet~\cite{imagenet_cvpr09} containing $1,000$ categories and $1.2$ million images on SSD. 
The other set containing $50,000$ labeled images is used as validation data.

\section{Configuration of High Performance Training System}
\label{sec:conf_perf}

We study configuration in three incremental steps, starting from a single GPU, then expanding our benchmarking to multiple GPUs, and finally to distributed nodes where each node consists of multi-GPUs.
Each of these three steps focuses on analyzing one system configuration.  

In the single GPU study, we analyze how the mini-batch size $X_{mini}$ can be decided to achieve fast training speed.
Most prior studies only consider tuning $X_{mini}$ algorithmically, that is, selecting a size that can achieve fast convergence.
However, taking the minimum number of epochs to reach convergence does not directly translate to shortest training time. 
In Section~\ref{ssec:singlegpu} we provide system analysis to determine $X_{mini}$ and solve optimized mini-batch selection with integer linear programming.

As multiple GPUs are employed to conduct training, data moving is the major bottleneck, which caps the speedup performance according to Amdahl's law.
Therefore, to be cost-effective, we should not use more GPUs when speedup improvement has saturated.
Section~\ref{ssec:multigpu} presents a systematic procedure to estimate an effective number of GPUs $G$.

When training is conducted in a distributed environment, we further study communication overhead.
Section~\ref{ssec:multihost} depicts the distributed training process and provides a lemma to estimate the required number of parameter servers in a cost-effective system configuration.

\subsection{Training on single GPU instance}
\label{ssec:singlegpu} 
In this section, we first point out the common performance pitfalls in designing neural networks. We illustrate that the setting of mini-batch size is the primary factor that determines training speed.  We then formulate selecting the mini-batch size $X_{mini}$ as an optimization problem and provide a procedure to solve for $X_{mini}$ that can achieve fastest training speed.

\subsubsection{Identifying System Issues}

Most neural networks are initially designed according to some heuristics. Researchers may not have the full picture about their model's feasibility, convergence quality, and prediction quality unless they conducted some experiments. During the experimental process, various hyper-parameter values may be tested exhaustively by a trial-and-error process. According to our own experience, it is typically unknown at the beginning to know how long it would take to run a round of training job, let alone configure a cost-effective system that can maximize training speed.
A suboptimal system configuration can lead to excessive execution time because of encountering the following issues:

\begin{itemize}
\item {\em Shortage of GPU memory space}. A GPU cannot commence computation without the data, including model parameters, gradients, computation workspace, etc, being loaded into GPU memory.
A neural network designed without system knowledge may require more memory capacity than available memory.  This excessive memory use may cause unnecessary thrashing and prolong training time.

\item {\em Ineffective tradeoff between speed and memory}. Deep learning frameworks may execute operations of a training task by using different algorithms, which have different speed and memory-use trade-offs.  
The selection of using which algorithm is a layer-dependent decision. The selection factors include input data size, layer parameters, mini-batch size, and available GPU memory space.
Consider the convolution operation as an example. An FFT-based algorithm runs faster than a GEMM-based one but it requires more memory. 
The training speed may be degraded when a large $X_{mini}$ exhausts memory capacity in order to run a faster FFT-based algorithm.
Thus, when tuning factors mentioned above, we should consider the impact on memory consumption because the memory budget affects the selection of algorithm.
\end{itemize}

Both training convergence and training speed can be decided by mini-batch size. In other words, selecting a good mini-batch size, one must examine from both the algorithmic and system aspects. From the algorithmic aspect, the mini-batch size is suggested to be larger than the number of output classes and a mini-batch contains at least one sample from each class~\cite{hinton2010practical}. The diversified training data leads to more stable convergence. From the system aspect, a proper mini-batch size helps to improve the parallelism inside GPU and enables the faster implementation of an operator. Based on the suggested mini-batch size considering the algorithmic aspect, we introduce the system aspect into deciding $X_{mini}$.

\subsubsection{Choosing Convolution Algorithms}

There are two time-consuming operations in deep learning: matrix multiplication and convolution.  Parallelizing matrix multiplication is rather straightforward, whereas speeding up convolution involves memory and speed trade-off.
\begin{table}[t]
\centering
\caption{Comparison of memory usage of convolution layers in AlexNet between FFT and GEMM}
\label{memoryconsumption}
\begin{tabular}{|c|c|c|}
\hline
Layer & Parameters & FFT/ \\
      & ($X_{mini}$,$B_{i}$,$H_{i}$,$B_{i+1}$,$H_{i+1}$,$D_{i}$,$D_{i+1}$,$F$) & GEMM\\ \hline
$1st$ & $(128,224,224,55,55,3,96,11)$ & $11.6$x \\ \hline
$2nd$ & $(128,27,27,27,27,96,256,5)$ & $1.6$x \\ \hline
$3rd$ & $(128,13,13,13,13,256,384,3)$ & $2.3$x \\ \hline
$4th$ & $(128,13,13,13,13,384,384,3)$ & $2.7$x \\ \hline
$5th$ & $(128,13,13,13,13,384,256,3)$ & $2.3$x \\ \hline
\end{tabular}
\end{table}
\begin{figure}[!ht]
\centering
\vspace{-1em}
\includegraphics[width=0.45\textwidth]{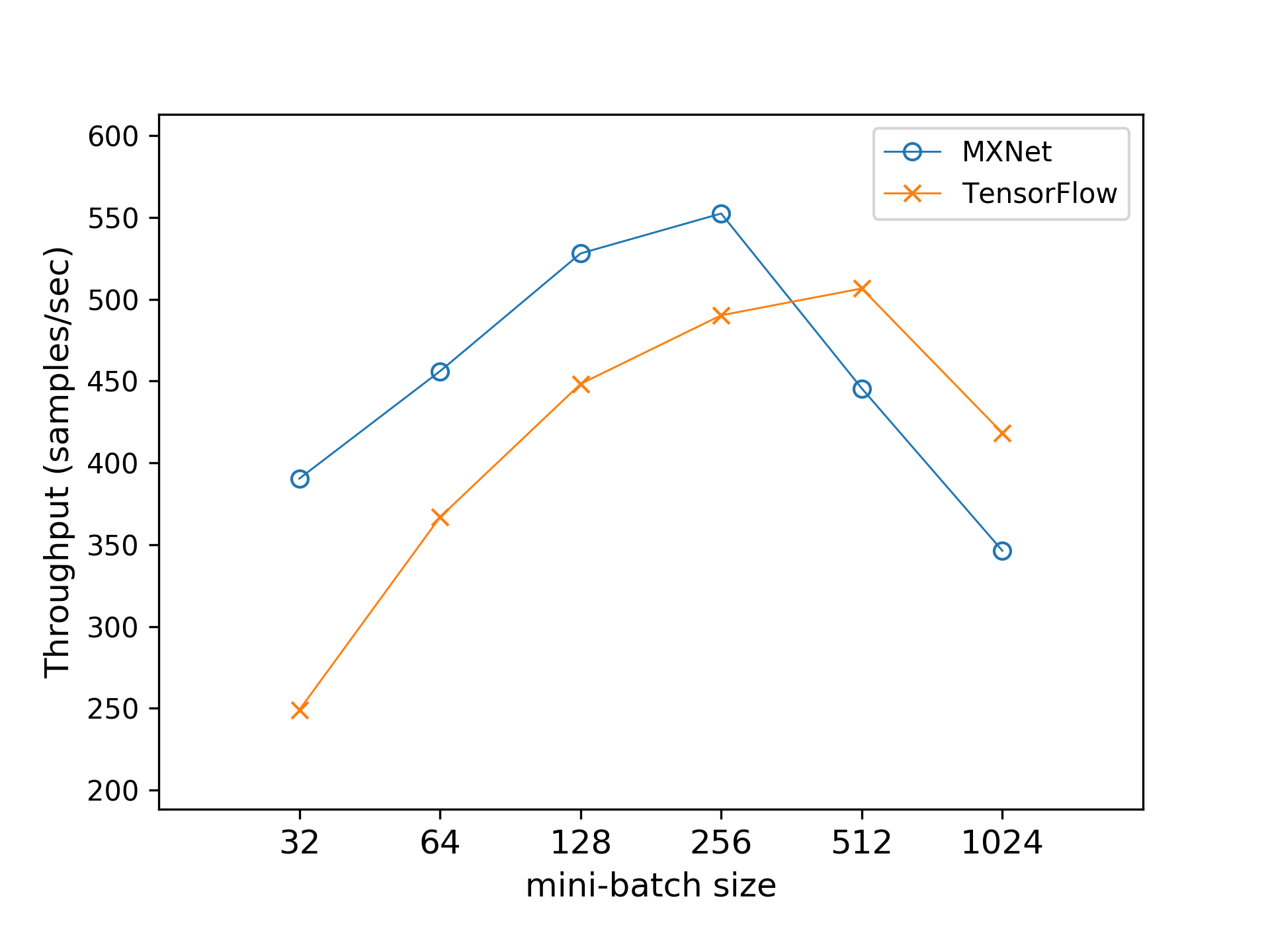}
\vspace{-2em}
\caption{Performance impact of mini-batch size}
\label{fig:exp1_throughtput_mini-batch}
\vspace{-1em}
\end{figure}
Two representative convolution algorithms are GEMM based~\cite{cudnn14} and FFT based~\cite{Mathieu:2013wa}.
GEMM-based algorithms converts convolution to a matrix multiplication, which can be slow but the up side is that it requires less memory space. 
FFT-based algorithms run faster than GEMM-based by using efficient matrix multiplication and reducing the number of floating point operations.
However, FFT-based algorithms demand substantially more memory as the filters are padded to be the same size as the input.
In addition, FFT-based algorithms require extra memory space for feature mapping on domain transformation.   
Table~\ref{memoryconsumption} shows five convolution layers of AlexNet and their memory-usage ratios of FFT over GEMM given mini-batch size $128$. The memory space required by the first layer with FFT is $11.6$ times of that required by GEMM. (The parameters $B_{i} \times H_{i}$ and $B_{i+1} \times H_{i+1}$ represent the number of pixels of the inputs and outputs at the $i^{th}$ layer, respectively. Similarly, the parameters $D_{i}$ and $D_{i+1}$ represent the depths of the inputs and outputs at the $i^{th}$, respectively. The parameter $F$ represents the size of filters.) 

To further understand the impact of $X_{mini}$, we experimented with MXNet and TensorFlow, and plot system throughout ($y$-$axis$) versus $X_{mini}$ ($x$-$axis$) in Figure~\ref{fig:exp1_throughtput_mini-batch}. Although different frameworks may yield different throughputs, the trend remains the same, that is, the system throughput degrades once after $X_{mini}$ reaches a threshold.  
The reason why the throughput drops is that MXNet and TensorFlow choose to run a slower convolution algorithm due to the constrained free memory caused by the increased $X_{mini}$. How to determine the optimal $X_{mini}$?  We next formulate the problem of determining $X_{mini}$ as an optimization problem.

\subsubsection{Optimizing Mini-batch Size}
\label{apdx:opt}

In order to formulate the problem of determining $X_{mini}$, we first define a memory constraint $M_{bound}$, which is built into the later optimization formulas for $X_{mini}$. During our formulation, most of the symbols follow in the same fashion of \cite{CS231CNN}. 

\vspace{0.5em}
\noindent \underline{Deriving $M_{bound}$}.
\vspace{0.5em}

We assume that a CNN such as AlexNet \cite{NIPS2012_alexnet} consists of two major components: feature extraction and classification. Further, we assume that the feature extraction part comprises of $n$ layers where stacked convolution layers are optionally followed by pooling layers, and the classification part consists of $m$ fully-connected layers. We use $B_{i} \times H_{i} \times D_{i}$ and $B_{i+1} \times H_{i+1} \times D_{i+1}$ where $i \in \{0,1,\ldots,n \}$ to represent the sizes of inputs and outputs of convolution layers (or pooling layers), respectively. In particular, the size $B_{0} \times H_{0} \times D_{0}$ represents the size of input data. If we take training AlexNet on the ImageNet \cite{imagenet_cvpr09} as the example, $B_{0} \times H_{0} \times D_{0}$ is equal to $224 \times 224 \times 3$. For the $i^{th}$ layer of convolution and pooling layers, we denote its spatial extent (i.e. the size of filters) as $F_{i}$, its stride as $S_{i}$, its amount of padding as $P_{i}$, and its number of filters as $K_{i}$. Please note that if the $i^{th}$ layer is a pooling layer, its $K_{i}$ is equal to zero, i.e. $K_{i}=0$. Thus, the inputs and outputs in the feature extraction part have the following relations:     
\begin{equation}\label{jason_re1}
\begin{aligned}
&B_{i+1}=(B_{i}-F_{i+1}+2P_{i+1})/S_{i+1}+1, \\
&H_{i+1}=(H_{i}-F_{i+1}+2P_{i+1})/S_{i+1}+1, and\\
&D_{i+1}=
\begin{cases}
K_{i+1}, & \mbox{if }(i+1)^{th} \mbox{ layer is convolution layer} \\
D_{i},   & \mbox{if }(i+1)^{th} \mbox{ layer is pooling layer}
\end{cases}
.
\end{aligned}
\end{equation}

The memory allocated for the feature extraction part of CNNs includes the input data, outputs (i.e. feature maps) of all the layers, model parameters, and gradients. We assume that all the values are stored by using single precision floating point ($32$bits). Based on the aforementioned notations and \Cref{jason_re1}, the memory usage for the input data and outputs of all layers in the feature extraction part can be calculated as follows:
\begin{equation}\label{mem:fm}
M_{FM} = \sum_{i=0}^{n} B_{i} \times H_{i} \times D_{i} \times X_{mini} \times 32 \,. 
\end{equation}
Regarding the model parameters, there are two kinds of parameters: weights and biases. Though the biases are often omitted for simplicity in the literature, we take them into account here in order to estimate the memory usage precisely. Besides, we assume that the size of the gradients is twice as the size of the model parameters \footnote{For each training instance, we need to store the gradients of all model parameters. The aggregated gradients of all model parameters are also required for a specific batch.}. Thus, we can derive the memory usage for the model parameters and their related gradients by the following equation:
\begin{equation}\label{mem:mp}
\begin{aligned}
M_{MP} = 
& \sum_{i=1}^{n} F_{i} \times F_{i} \times D_{i} \times K_{i} \times 3 \times 32 \,+ &(weights) \\
& \sum_{i=1}^{n} K_{i} \times 3 \times 32 \, &(biases)\,.
\end{aligned}
\end{equation}

Furthermore, the memory allocated for the classification part of CNNs contains the outputs of all neurons and model parameters. We use $L_{j}$ where $j \in \{1,\ldots,m\}$ to denote the number of neurons at $j^{th}$ layer. Again, we make the same assumption that the size of the gradients is twice as the size of the model parameters. Therefore, the memory usage for the classification part of CNNs is as follows:   
\begin{equation}\label{mem:c}
\begin{aligned}
M_{C} = 
& \sum_{j=1}^{m} L_{j} \times 32 \,+ &(outputs) \\
& \sum_{j=1}^{m-1} L_{j} \times L_{j+1} \times 3 \times 32 \,+ &(weights) \\
& (m-1) \times 3 \times 32 \, &(biases)\,.
\end{aligned}
\end{equation}
According to \Crefrange{mem:fm}{mem:c}, the memory constraint $M_{bound}$ can be approximately determined by the following equation:
\begin{equation}\label{mem:bound}
M_{bound} = M_{GPU} - M_{FM} - M_{MP} - M_{C},
\end{equation}
where $M_{GPU}$ is the total memory of a GPU in terms of bits.

\vspace{0.5em}
\noindent \underline{Deriving $X_{mini}$}.
\vspace{0.5em}

Assuming that there are $p$ kinds of convolution algorithms, and $q$ layers in the CNN. (In the case that we have illustrated so far, $p = 2$. Other choices of convolution algorithms can be Winograd minimal convolution algorithm~\cite{lavin2016fast}, Strassen algorithm~\cite{cong2014minimizing}, fbfft~\cite{vasilache2014fast}, etc.) The parameter $x_{k,l} \in \{0,1\}$ represents whether the $k^{th}$ layer uses the $l^{th}$ convolution algorithm or not. When $x_{k,l}$ is evaluated to $1$, it means that the $k^{th}$ layer uses the $l^{th}$ algorithm to compute convolution. The value $T_{k,l}$ is the time consumption at the $k^{th}$ layer for the $l^{th}$ algorithm. The value $M_{k,l}$ is the memory consumption at the $k^{th}$ layer for the $l^{th}$ algorithm. Thus, the problem of determining $X_{mini}$ can be formulated an optimization problem as follows:
\begin{equation}\label{jason_eq1}
\begin{aligned}
&min \sum_{k=1}^{q} \sum_{l=1}^{p} x_{k,l} \times T_{k,l} \\
&s.t. \sum_{k=1}^{q} \sum_{l=1}^{p} x_{k,l} \times M_{k,l} \le M_{bound} \quad and \\
& \forall k \, \sum_{l=1}^{p} x_{k,l} = 1,
\end{aligned}
\end{equation}
where the $M_{bound}$ is derived from \Cref{mem:bound}. 

Obviously, \Cref{jason_eq1} is an integer linear programming (ILP) problem \cite{nemhauser1988integer}, which is NP-hard. However, there are several off-the-shelf heuristic methods and libraries (e.g. GLPK \cite{GLPK}) for solving ILP problems.
Given a range of mini-batch sizes that can attain good accuracy, 
we can derive the estimated training time for each mini-batch size by solving \Cref{jason_eq1}. The mini-batch size which leads to the minimal training time is then the suggested $X_{mini}$.

\subsubsection{Refining Model for Speed}
\begin{figure}[t]
\centering
\includegraphics[width=0.45\textwidth]{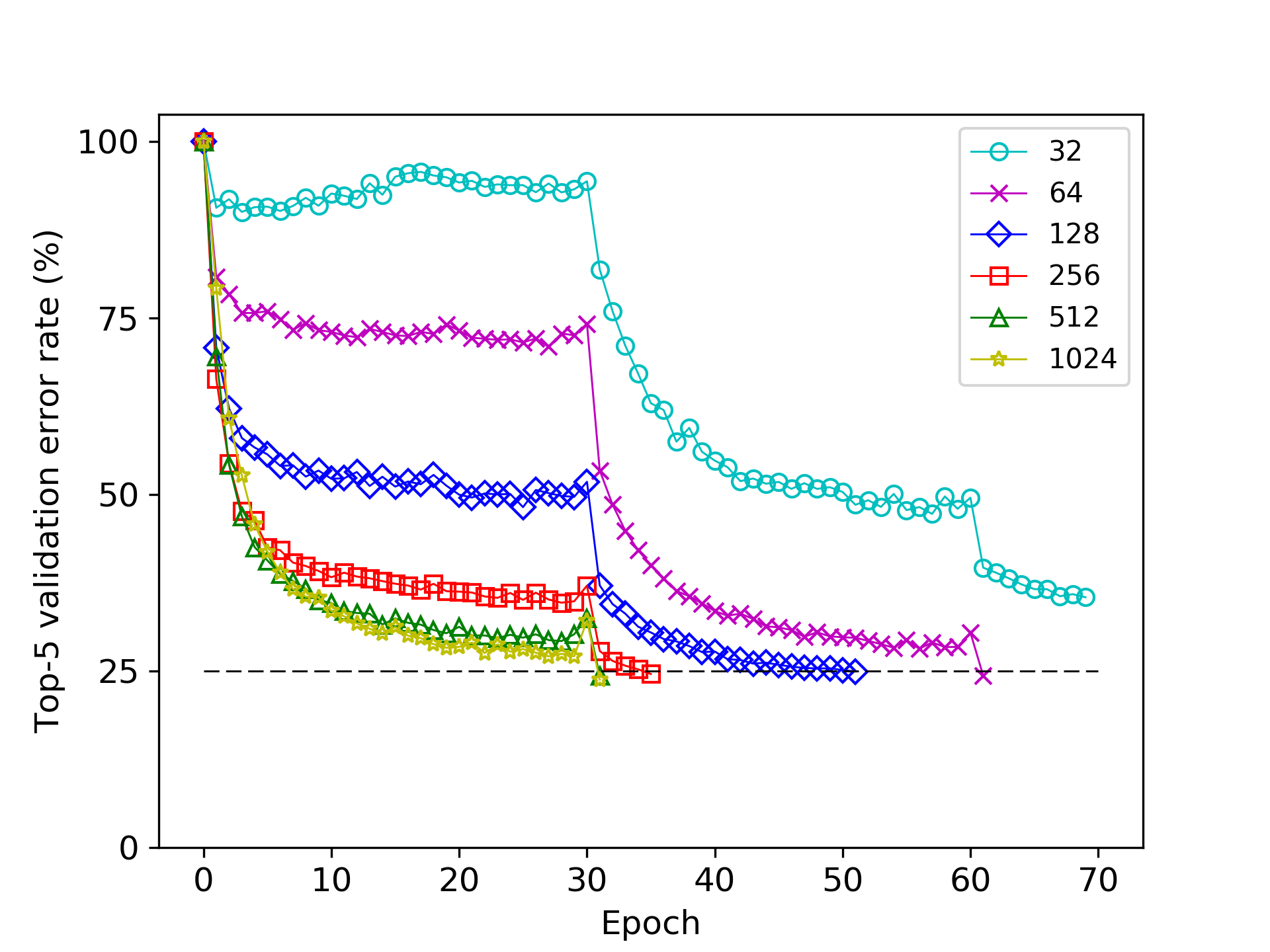}
\vspace{-1em}
\caption{Learning curves of different mini-batch sizes with respect to number of epochs}
\label{fig:exp1_epoch_error}
\vspace{-1em}
\end{figure}

This far, we assume that a CNN model is given to determine $X_{mini}$ and layer-dependent convolution algorithms to maximize training speed.  We can make two further adjustments:
\begin{itemize}
\item {\em Permit $X_{mini}$ reduction}.
The researchers may need to compromise on smaller mini-batch size if the target one is not feasible or does not deliver acceptable performance under the constraint of GPU memory size. Ghadimi et al.~\cite{ghadimi2013stochastic} shows that the convergence rate of SGD on a non-convex function is bounded by $O(1/\sqrt[]{K})$, where $K$ is the number of samples seen, i.e., mini-batch size. 
It can be interpreted that a range of mini-batch sizes can deliver similar convergence quality. In Figure~\ref{fig:exp1_epoch_error}, the $x$-axis depicts the epoch number and the $y$-axis depicts the top-$5$ validation error rate\footnotemark.
\footnotetext{AlexNet achieved $18.2$\% top-5 error rate in in the ILSVRC-2012 competition, whereas we obtained $21$\% in our experiments. This is because we did not perform all the tricks for data augmentation and fine-tuning. We choose $25$\% as the termination criterion to demonstrate convergence behavior when mini-batch sizes are different.}
The figure shows that indeed a range of mini-batch sizes enjoy similar convergence quality. Therefore, we could reduce $X_{mini}$ to increase $M_{bound}$ to permit more memory space to run a faster convolution execution to achieve overall speedup.

\item {\em Permit model adjustment}. Suppose that the constrained space of memory prevents us from running a faster algorithm. We could adjust the CNN model to free up some memory.  For instance, if the $i^{th}$ layer can be sped up ten times and the $j^{th}$ only twice.  To accommodate running a faster algorithm for the $i^{th}$ layer, we could adjust both layers to e.g., use a larger stride or memory-efficient filters.
\end{itemize}

\subsection{Scale with Multiple GPUs}
\label{ssec:multigpu}
When one GPU cannot handle the training task timely, employing multiple GPUs is the next logical step to share the workload and achieve speedup.  When $G$ GPUs are used and the maximal $100\%$ efficiency is achieved, the speedup is $G$ times.  Let $\alpha$ denote the system efficiency between $0\%$ and $100\%$.
Lemma~\ref{th:alpha} provides the estimated efficiency given $G$ GPUs.

\begin{lemma}
\label{th:alpha}
Let $T$ denote the total training time, where $T$ can be divided into computation time $T_{C}$ and overhead $T_{O}$. Let $R_O$ denote the ratio of overhead or $R_O = T_O / T_C$. Suppose the desired efficiency of the system is $\alpha$, where $\alpha \le 100\%$. The efficiency can be estimated as
\begin{displaymath}{\alpha = \frac{1 + R_{O}}{1 + GR_{O}}}.\end{displaymath}
\end{lemma}

\begin{proof}
Details of the proof is documented in Appendix~\ref{apdx:proof}.
\end{proof}

Lemma~\ref{th:alpha} can be used to estimate system efficiency given $R_O$ and $G$, and also can be used to estimate the acceptable $R_O$ given $\alpha$ and $G$.  For example, given four GPUs and target efficiency $\alpha = 80\%$, the ratio of overhead that cannot be hidden behind computation must not exceed $9\%$. 

To estimate $R_{O}$, a practitioner can quickly profile the training program for a couple of epochs.
Some frameworks such as MXNet and TensorFlow provide the capability to visualize the execution of a training task, which can be used to derive $R_{O}$. 
If a computation framework is not equipped with a profiling tool, one can visualize program execution  using \textbf{nvprof}\footnotemark. \footnotetext{\textbf{nvprof} only profiles GPU activities, so the CPU activities cannot be analyzed.} Suppose a practitioner is asked to make $3x$ speedup of a training task, and she measures $R_{O} = 10\%$. According to the lemma, she can configure a $4$ GPU system to achieve the performance objective.

\begin{figure}[!ht]
\centering
\centerline{\includegraphics[width=0.45\textwidth]{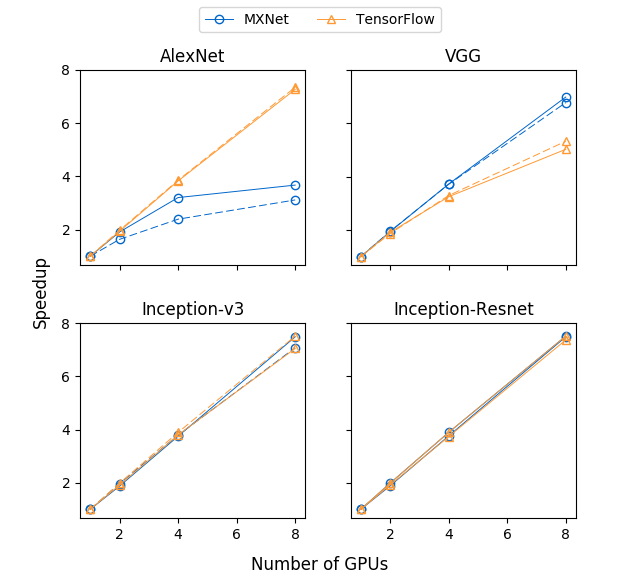}}
\vspace{-1em}
\caption{Comparison of speedup (dotted-line: estimated, solid-line: actual)}
\label{fig:exp3_predict}
\vspace{-2em}
\end{figure}

To evaluate Lemma~\ref{th:alpha}, we conduct the training on four neural networks to compare the estimated speedup with actual speedup. Though the estimated $R_O$ is a constant and in real-time overheads could be stochastic, Figure~\ref{fig:exp3_predict}  shows that in all cases the estimated speedup matches the the actual speedup.
Therefore, the lemma can be used to estimate the performance gain of using $G$ GPUs and devise a cost-effective training plan including system configuration and parameter settings.

The overall speedup can be improved by reducing computation overheads. We conclude this subsection by providing two overhead reduction suggestions.  
\begin{itemize}
\item {\em Data transfer pipelining}.
Low throughput of feeding training data is a major bottleneck that degrades the multi-GPU training performance as the demand for bus bandwidth for loading data grows with the number of GPUs. 
Pipelining data loading (I/O) with computation is the effective way to reduce the overhead brought by data preparation. The impact of disk I/O can be further alleviated by using better disk or reducing expensive file operations like seek. Modern frameworks such as TensorFlow and MXNet provide the way to rearrange training samples so that the data can be read in sequentially. The load for decoding and augmenting training data may cause extreme high CPU usage and drags the performance of data provision. The computation intensive jobs should be avoided on CPUs.

\item {\em Peer-to-peer parameter updates}. Synchronizing parameter updates among GPUs, as indicated in step $6$ in Figure~\ref{fig:sec3_node_arch}, is another common bottleneck in multi-GPU training environment. A naive implementation is to keep the latest model at main memory, transfer the latest copy to GPUs at the beginning of batch processing, and aggregate updates from all GPUs.  It leads to bus contention and huge data load between main memory and GPUs under CUDA programming model.
To alleviate the hot spot issue, the weight updates can be completed via GPU high-speed DMA if GPU supports peer-to-peer transfer.
\end{itemize}

If multiple GPUs with low computing overhead still cannot meet the desired performance, distributed training is the option you can consider. We'll discuss the topic in the next section.

\subsection{Distributed Training}
\label{ssec:multihost}

Distributed training has become increasingly important because of the growth of dataset size and model complexity. To effectively orchestrate multiple machines for a training task, the system must provide a way to manage the globally shared model parameters. The parameter server architecture, i.e., a cluster of machines to manage parameters, is widely-used to reduce I/O latency for handling parameter updates~\cite{PLDA+,Li:2014tt}. As shown in Figure~\ref{fig:sec3_node_arch}, parameter servers maintain latest parameter values and serve all workers. The workers retrieve updated parameters from the cluster, complete computation, and then push updates back to the cluster of parameter servers. 

Parameter updates can be performed either synchronously or asynchronously. Employing synchronous updates ensures consistency but suffers from the performance dragger issue. Updating parameters asynchronously gains training speed and may not significantly affect training accuracy according to prior studies \cite{dean2012large}. When I/Os can be performed asynchronously, fetching and updating parameters can be hidden behind computation and hence computation overhead can be mitigated. We assume that an asynchronous update policy is employed.

Let $N_{ps}$ denote the number of parameter servers.
How many parameter servers should be configured to hide the computation overhead?  
We select $N_{ps}$ when $N_{ps} +1$ can no longer speed up the training task.
Before we prove our lemma that derives the most effective $N_{ps}$, we enumerate two desired subgoals or conditions.  

The first subgoal is that the computation duration of a worker should be longer than its communication time with the parameter cluster.  In other words, the I/O time between a worker thread and its designated parameter servers is shorter than the computation time of that worker.  This condition allows parameters being pre-fetched before a new round of computation commences.  Therefore, the I/O overhead can be hidden behind computation.  The second subgoal is to distribute parameter-update workload evenly among parameter servers. We assume a dynamic load-balancing policy (e.g., \cite{2DBubbleUp}) can be employed to distribute parameter retrieval and update workload almost evenly among $N_{ps}$ servers.

\begin{lemma}
\label{def:n_ps}
Given a round of GPU computation time $T_{C}$ on a worker, number of workers $N_{w}$, network bandwidth $B_{ps}$, and parameter size $S_{p}$, the minimum number of parameter servers $N_{ps}$ required to mask communication I/Os is
\begin{displaymath}
{N_{ps} \simeq \left \lceil{ \frac{2  S_{P}  N{w}}{B_{ps} T_{C}}}\right \rceil}.
\end{displaymath}
\end{lemma}

\begin{proof}
The total size of communication I/O load generated in a round of pull to and push from parameter servers is $2 \times S_{p} \times N_{w}$.  Given that the I/O bandwidth is $N_{ps}$ and the load evenly distributed among $N_{ps}$ servers, the communication time can be written as $2 \times S_{p} \times  N_{w} / N_{ps} \times B_{ps}$. The ideal pipeline case \cite{PLDA+} is when the I/O time can be hidden behind computation time. Therefore, the I/O time must be smaller than or equal to the computation time $T_C$. (The parameter update time on a parameter server is ignored because that time is relative small comparing with network transmission time.)  We can write the constraint to be

\begin{equation}
\begin{aligned}
{T_{C} \geqslant \frac{2 \times S_{p} \times N{w}}{N_{ps} \times B_{ps}}}.
\end{aligned}
\end{equation}

\enlargethispage{-\baselineskip}
Isolating $N_{ps}$ on the left-hand side of the equation, we obtain

\begin{equation}
\begin{aligned}
{N_{ps} \geqslant \frac{2 S_{p} N{w}}{T_{C} B_{ps}}}.
\end{aligned}
\end{equation}

\end{proof}

Lemma~\ref{def:n_ps} suggests a back-of-the-envelop estimate on $N_{ps}$ given two ideal conditions. When the conditions do not hold, more parameter servers should be employed to be able to mask I/O overhead. Three measures are recommended:

\begin{enumerate}
\item {\em Increase $T_C$}. When workload cannot be evenly distributed, the computation time should be longer to mask most I/Os.  Therefore, a good strategy is to maintain a large $T_C$.  In other words, having a larger mini-batch size when the memory capacity permits is helpful. Besides, a larger mini-batch leads to less number of parameter updates and improves overall performance.

\item {\em Improve $B_{ps}$}. Increasing network bandwidth can reduce I/O time. 
Insufficient network bandwidth of the communication channel may throttle the training performance. Take AlexNet as an example, pushing parameter updates produces around $180 MB$ network traffic, which exceeds the capacity of commonly used $1 Gbit$ Ethernet. Thus, high speed networking is highly recommended when applying distributed training.

\item {\em Balance workload}.  Prior works \cite{2DBubbleUp, PLDA+} propose effective data placement methods to balance dynamic workload. Such load balancing schemes can avoid I/O bottlenecks, and lead to overall overhead reduction.
\end{enumerate}

\vspace{0.5em}
\section{Concluding Remarks}
In this work, we investigated typical deep learning frameworks running on representative deep learning models and datasets.
From analyses, we studied the computation bottlenecks in single-GPU, multi-GPU and distributed configurations.
Furthermore, we derived the back-of-the-envelope estimation for the GPU number to configure a training system, given a budget or deadline.
Finally, for distributed training, we suggested a formula for estimating the number of parameter servers to be configured to reduce communication overhead.

AlphaGo showed that more training data can only be helpful towards improving machine intelligence and competitiveness. Recently, Residual Neural Networks \cite{he2016deep, inception-v4} shows that in both theory and practice, more layers of neural networks correlates to a higher achieved accuracy by a trained classifier. At a 2016 machine learning workshop \cite{andrew2016}, Andrew Ng presented that the traditional biases and variance tradeoff have not appeared in training large-scale deep architectures. In other words, the larger the scale, the better suited the architecture is for improving the intelligence of a ``machine''.

This ``larger the better'' conclusion certainly demands that database and machine learning communities devise data management and data mining systems that can handle an ever increasing workload. We foresee that not only will algorithmic research continue flourishing, but system research and development will as well. Already we have seen that GPU vendors are enhancing distributed GPU implementations. Advances in interconnected technology and implementation will help reduce both I/O overhead in data loading and in parameter updates.

In this work, we provided practical guidelines to facilitate practitioners the configuration of a system to speed up training performance. 
Our future work will focus on effectively managing such large-scale training systems to
achieve both high accuracy and cost-effectiveness in three specific areas:

\begin{itemize}
\item {\em Flexibility}. Prior work \cite{Anonymous:NXi9MKDQ} provided a flexibility to work with any compatible open-source frameworks.  
For example, we expect to simultaneously work with multiple frameworks such as MXNet and TensorFlow to complete a large-scale training task running on Azure, AWS, GCE, and other available commercial clouds.
\item {\em Scalability and elasticity}. In addition to the parameter estimation performed in this work, we will research dynamic schemes to adjust  
allocation and scheduling parameters according to the dynamic workload nature of distributed systems.
\item {\em Ease of management}. We plan to devise tools with the good user experience for monitoring and managing the training system.
\end{itemize}

\appendix
\section{Appendices}

\subsection{Proof of Lemma~\ref{th:alpha}}
\label{apdx:proof}
According to Amdahl's law, given $G$ GPUs and the fraction of the execution time of the task that can be parallelized $P$, the theoretical speedup is $\frac{1}{(1-P)+\frac{P}{G}}$. The maximum speedup $G$ can not be achieved if there are parts cannot be parallelized. Thus:
\begin{equation}
\label{eq:amdahl}
\alpha G = \frac{1}{(1-P)+\frac{P}{G}}
\end{equation}
$P$ can be expressed as:
\begin{equation}
 P = \frac{T_{C}}{T} = \frac{T_{C}}{T_{C}+T_{O}}
\end{equation}
Substituting $P$ into equation~\ref{eq:amdahl} yields:
\begin{equation}
\frac{T_{O}}{T_{C}} = \frac{1 - \alpha}{\alpha G - 1}
\end{equation}
Then:
\begin{equation}
\label{eq:amro}
R_{O} = \frac{1 - \alpha}{\alpha G - 1}
\end{equation}

By rearranging equation~\ref{eq:amro}, $\alpha$ can be expressed in terms of $G$ and $R_{O}$ as follows:
\begin{equation}
\label{eq:alpha}
{\alpha = \frac{1 + R_{O}}{1 + GR_{O}}}
\end{equation}

\linespread{1.12}
\bibliographystyle{ACM-Reference-Format}
\bibliography{reference} 

\end{document}